\documentclass[copyright,creativecommons]{eptcs}
\providecommand{\event}{TARK 2023}
\usepackage{iftex}
\usepackage{underscore}         % Only needed if you use pdflatex.
\usepackage[T1]{fontenc} 
\usepackage[english]{babel}
\usepackage{multirow,bigdelim}
\usepackage{centernot}

\usepackage{tikz}
\usetikzlibrary{arrows,chains,matrix,positioning,scopes}
\makeatletter
\tikzset{join/.code=\tikzset{after node path={%
\ifx\tikzchainprevious\pgfutil@empty\else(\tikzchainprevious)%
edge[every join]#1(\tikzchaincurrent)\fi}}}
\makeatother
\tikzset{>=stealth',every on chain/.append style={join},
         every join/.style={->}}
\tikzstyle{labeled}=[execute at begin node=$\scriptstyle,
   execute at end node=$]

\usepackage{booktabs} % For formal tables
\usepackage[ruled]{algorithm2e} % For algorithms

\SetAlFnt{\small}
\SetAlCapFnt{\small}
\SetAlCapNameFnt{\small}
\SetAlCapHSkip{0pt}
\IncMargin{-\parindent}

\usepackage{amsmath, amsthm,amsfonts}
\usepackage{mathtools}
\usepackage{thm-restate}
\usepackage{color}
\usepackage{lscape}
\usepackage{soul}
\usepackage{bm, bbm}
\usepackage{enumerate}
\usepackage{tabto}

\newcommand{\AutoAdjust}[3]{{ \mathchoice{ \left #1 #2  \right #3}{#1 #2 #3}{#1 #2 #3}{#1 #2 #3} }}
\newcommand{\Xcomment}[1]{{}}

\newcommand{\InBrackets}[1]{\AutoAdjust{[}{#1}{]}}% {\left[{#1}\right]}
\newcommand{\Ex}[2][]{\operatorname{\mathbbm E}_{#1}\InBrackets{#2}}

\newtheorem{theorem}{Theorem}[section]
\newtheorem{proposition}[theorem]{Proposition}
\newtheorem{lemma}[theorem]{Lemma}
\newtheorem{definition}[theorem]{Definition}
\newtheorem{claim}[theorem]{Claim}
\newtheorem{corollary}[theorem]{Corollary}

\newtheorem{example}[theorem]{Example}

\DeclareMathOperator*{\argmin}{arg\,min}
\DeclareMathOperator*{\argmax}{arg\,max}

\newcommand{\noaccents}[1]{#1}
\newcommand{\newagentvar}[3][\noaccents]{%
\expandafter\newcommand\expandafter{\csname #2\endcsname}{#1{#3}}%
\expandafter\newcommand\expandafter{\csname #2s\endcsname}{#1{\boldsymbol{#3}}}%
\expandafter\newcommand\expandafter{\csname #2smi\endcsname}[1][i]{#1{\boldsymbol{#3}}_{-##1}}%
\expandafter\newcommand\expandafter{\csname #2i\endcsname}[1][i]{#1{#3}_{##1}}%
\expandafter\newcommand\expandafter{\csname #2ith\endcsname}[1][i]{#1{#3}_{(##1)}}%
}

\newcommand{\newvecagentvar}[3][\noaccents]{%
\expandafter\newcommand\expandafter{\csname #2\endcsname}{#1{\boldsymbol{#3}}}%
\expandafter\newcommand\expandafter{\csname #2s\endcsname}{#1{\boldsymbol{#3}}}%
\expandafter\newcommand\expandafter{\csname #2smi\endcsname}[1][i]{#1{\boldsymbol{#3}}_{-##1}}%
\expandafter\newcommand\expandafter{\csname #2i\endcsname}[1][i]{#1{\boldsymbol{#3}}_{##1}}%
\expandafter\newcommand\expandafter{\csname #2ith\endcsname}[1][i]{#1{#3}_{(##1)}}%
}

\newagentvar{val}{v}
\newagentvar{bid}{b}
\newagentvar{dist}{F}
\newagentvar{alloc}{x}
\newagentvar{util}{u}
\newagentvar{pay}{p}
\newagentvar{ability}{a}
\newagentvar{effort}{e}

\title{Optimal Mechanism Design for Agents with DSL Strategies: The Case of Sybil Attacks in Combinatorial Auctions}
\author{Yotam Gafni
\institute{Technion \\ Haifa, Israel}
\email{yotam.gafni@campus.technion.ac.il}
\and
Moshe Tennenholtz
\institute{Technion \\ Haifa, Israel}
\email{moshet@ie.technion.ac.il}
}
\def\titlerunning{DSL Reasoning under Uncertainty}
\def\authorrunning{Y. Gafni \& M. Tennenholtz}

\begin{document}

\maketitle

\begin{abstract}
     In robust decision making under uncertainty, a natural choice is to go with 
     safety (aka security) level strategies. However, in many important cases, most notably auctions, there is a large multitude of safety level strategies, thus making the choice unclear. 
     We consider two refined notions: 
     \begin{itemize}
     \item a term we call DSL (distinguishable safety level), and is based on the notion of ``discrimin'' \cite{discrimin}, which uses a pairwise comparison of actions while removing trivial equivalencies. This captures the fact that when comparing two actions an agent should not care about payoffs in situations where they lead to identical payoffs. 
     \item The well-known Leximin notion from social choice theory, which we apply for robust decision-making. In particular, the leximin is always DSL but not vice-versa \cite{discrimin}. 
     \end{itemize}
     We study the relations of these notions to other robust notions, and illustrate the results of their use in auctions and other settings. Economic design aims to maximize social welfare when facing self-motivated participants. In online environments, such as the Web, participants' incentives take a novel form originating from the lack of clear agent identity---the ability to create Sybil attacks, i.e., the ability of each participant to act using multiple identities. It is well-known that Sybil attacks are a major obstacle for welfare-maximization. 
Our main result proves that when DSL attackers face uncertainty over the auction's bids, the celebrated VCG mechanism is welfare-maximizing even under Sybil attacks. Altogether, our work shows a successful fundamental synergy between robustness under uncertainty, economic design, and agents' strategic manipulations in online multi-agent systems. 
\end{abstract}

\section{Introduction}

Consider an agent who needs to decide on her action in an environment consisting of other agents. In certain cases there is a uniquely defined optimal action for the agent, but in most  cases this ``agent perspective'' is an open challenge. Given the above, both AI and economics care about an adequate modeling of an agent, and its ramifications in a variety of multi agent contexts, for example, on social welfare.

We consider a notion for agent modeling we term DSL (Distinguishable Safety-Level). The notion was previously suggested in the context of constraint-satisfaction problems and fuzzy logic, and was termed ``discrimin'' \cite{discrimin}. In game theoretic settings, the notion was previously applied \cite{discriminBoolean} as a solution concept for bargaining in Boolean games \cite{booleanGames}. To the best of our knowledge, it was not previously considered in the context of auctions, voting, and more generally mechanism design, i.e., when considering the robustness of economic mechanisms' performance when facing strategic agents. 

There are two ways to think of the DSL solution concept, when applied to agent modeling. One is as a solution concept adapted to capture the behavioral phenomenon of the loss aversion cognitive bias in agents, particularly when probabilities over nature states are unknown. The other is as a form of robust strategy choice under uncertainty, that may be required in volatile and unpredictable environments that do not admit a stable Bayesian description. We show its usefulness in auctions. In the full version of this paper, we also study its behavior in other prominent strategic settings, such as voting. %\moshe{Need to re-visit if moved to appendix?} 
In our main result we consider the celebrated welfare maximizing VCG mechanism in combinatorial auctions setting, where it is known to fail under false name (aka Sybil) attacks. We show that DSL agents lead to optimal social welfare. 

%We use the rest of the introduction for discussion of the approach we offer, related background, as well as a careful overview of VCG under false-name attacks, which serves as the setting of our main result. 

\subsection{Reasoning under Uncertainty}

A classic distinction \cite{merrill1981strategic} separates reasoning under risk, where the actors are rational and there is a commonly known distribution about their environment  (also known as the stochastic or Bayesian setting), and reasoning under uncertainty, where the general structure of strategies and outcomes is known, but there is no probabilistic information about the environment. Moreover, even assumptions regarding actors' rationality or behavior characteristics may not be guaranteed %\yotam{Cite the discrimin boolean games paper?}
. For such cases, a robust or worst-case approach seems appropriate, and various notions exist to capture it. Ideally, a dominant strategy solution exists, but this is usually not the case (and indeed it is not the case in all the cases we analyze in this paper). A minimal robust notion is that of a safety level strategy, which uses a max-min approach over all possible outcomes given a strategy choice. However, though it yields interesting results in some cases \cite{aumann1985non, tennenholtz2001rational}, in many other cases it does not tell us much about what strategy to choose, in particular in auctions settings, %\moshe{need to re-visit if we change structure?} 
where we derive our most interesting results. As we see, this is because in auctions the natural safety level is $0$ (which happens when the bidder loses the auction), and any strategy that does not overbid guarantees it. It is thus hard to choose among these strategies without considering a more refined notion. 
Existing refined notions are the lexicographic max-min (originally defined in \cite{leximin}) and 
%\moshe{it is now in appendix?}
min-max regret \cite{savage1951theory}. We overview their comparison to the notion of DSL in Section~\ref{SEC:OTHER_CONCEPTS} and Appendix~\ref{app:minmax_regret}, respectively. 

%Recently, a minimal notion of Incentive Compatibility was applied to valuations reporting in the context of fair share allocation \cite{babaioff2022selfMaximizing}

\subsection{VCG, Sybil Attacks, and Welfare}

VCG is a well known mechanism which can be applied for combinatorial auctions. VCG has good qualities such as being dominant strategy incentive compatible and achieving optimal social welfare. However, under the possibility of false-name attacks \cite{yokoo2004effect}, it is no longer truthful. Coming up with other mechanisms does not solve the basic conundrum: In the full information settings, any false-name proof mechanism performs poorly in terms of welfare \cite{iwasaki2010worst}. 

A possible avenue to solving the issue is by limiting the discussed valuation classes. However, an example in \cite{lehmann2006combinatorial} shows that even when all bidders have sub-modular valuations, VCG is no longer dominant strategy incentive-compatible under false-name attacks. Notably though, even with this example, VCG still arrives at the socially optimal allocation, and in fact as \cite{alkalay2014false} show, this observation is true in general up to a constant with sub-modular (and near sub-modular) bidders. However, in the full version of our paper, we show an example %Example~\ref{ex:XOS_bad_welfare} 
where for the XOS valuation class, which extends the sub-modular class, there is such an attack so that VCG arrives at an arbitrarily sub-optimal allocation. 
The attack we describe is enabled by the full information settings. Without full information, the attack is risky for the attacker, since it could lead to negative utility, as the attacker overbids her true valuation. %(for example, the total stated value of the entire set of items given the two shill bids is $40$, while in truth it is $6\epsilon$). 
%Thus, with other nature bidders, it could be asked to pay $40$ and end up with negative utility.  
%\moshe{We must crystallize the fact there is nothinh know for obtaining optimal social welfare for general valuation functions, as I believe our claim we are the first to obtain that, right? Otherwise, it can be emphasized that this is our main result.   }

A useful approach, that can lead to better welfare guarantees than dominant strategy mechanism design, is Bayesian mechanism design. Assuming that the bidder distributions are common knowledge, recent work has shown that selling each item separately leads to good constant approximation welfare guarantees for XOS \cite{christodoulou2008bayesian} and sub-additive \cite{feldman2013simultaneous} valuations. Though the works do not explicitly consider false-name attacks, their constructions use the false-name-proof first and second price auctions to auction items separately, and so their results naturally extend to Bayesian false-name mechanism design.%it is not too hard to show that these results extend to when it is also possible to perform false-name attack, as we show in Proposition~\ref{prop:bayesian_vcg}. 

It is important to note, that many of the above positive results for welfare guarantees under false-name attack assume some form of risk-aversion; most importantly, that bidders do not overbid, i.e., they choose only strategies that are individually rational (under any possible nature state). This condition is equivalent to limiting the strategy space only to safety level strategies %\moshe{may be confusing, as safety-level strategies are particular strategies rather than constraints?}
(as in this case of combinatorial auctions, the safety level is 0). In \cite{gafni2020vcg} the authors do not make this assumption, but their positive welfare optimality results are limited as they only consider the homogeneous single-minded case with two items. We thus believe that it is natural to ask: Under our definition of DSL, which is a strong risk-aversion notion (compared, e.g., to the safety level strategy), what welfare guarantees can be obtained? Surprisingly, the answer is optimal, as we show in our main result in Theorem~\ref{thm:vcg}.

\subsection{Our Results}
In Section~\ref{sec:motivating_example} we formally define our solution concept, and apply it to the first-price and discrete first-price auctions. In Section~\ref{SEC:OTHER_CONCEPTS} (with the additional discussion of min-max regret in Appendix~\ref{app:minmax_regret}) we describe a hierarchy of solution concepts and their relations to the solution concept we introduce (\text{DSL}), as summarized in Figure~\ref{fig:notions_hierarchy}.  

\begin{figure}
\centering
\includegraphics[width=0.3\textwidth]{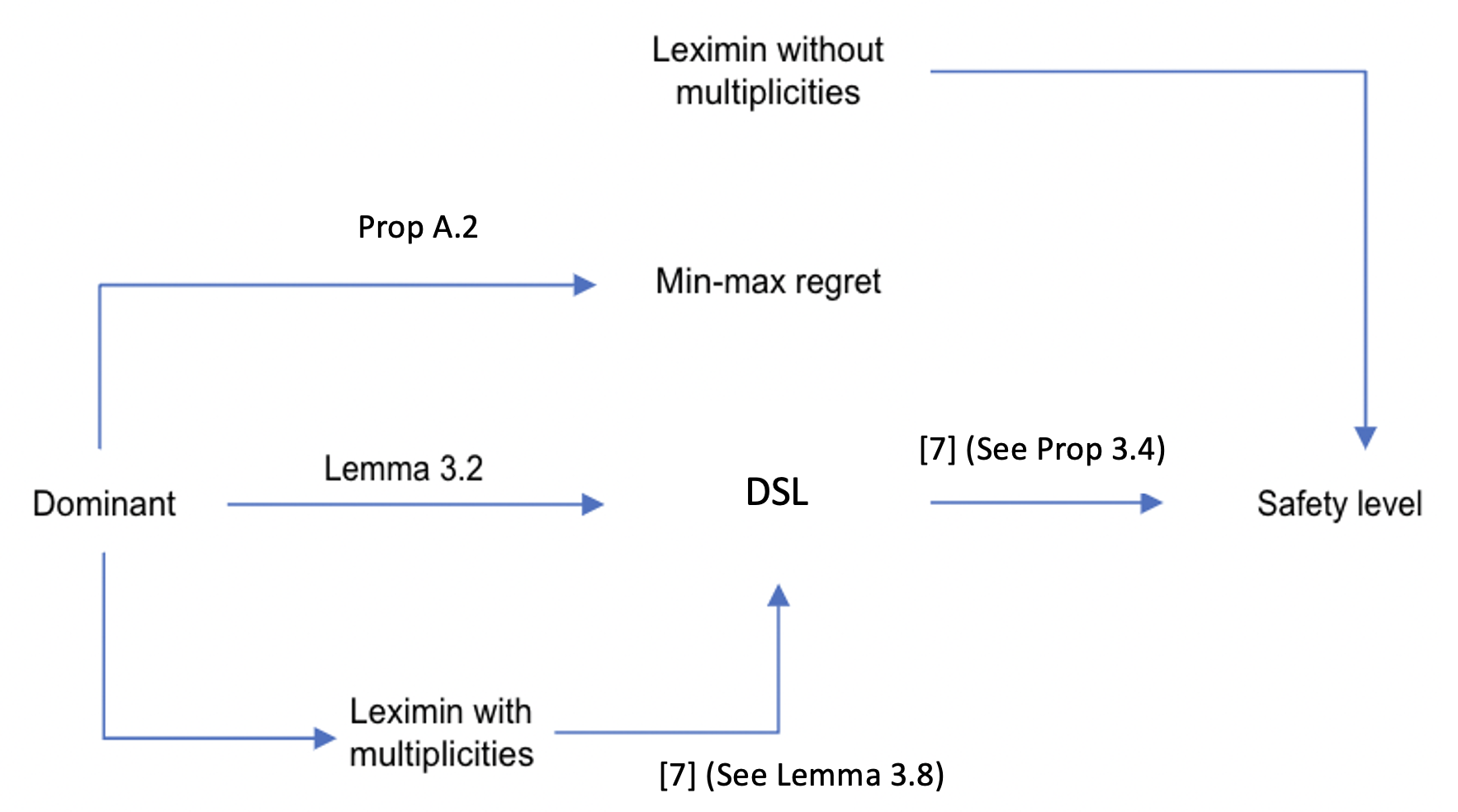}
\caption{Hierarchy for robust decision under uncertainty}
\label{fig:notions_hierarchy}
\end{figure}

%\begin{figure}
%\centering 
%\includegraphics[width=0.3\textwidth]{hierarchy_graph2.png}
%\caption{Hierarchy for robust decision under uncertainty} 
%\label{fig:notions_hierarchy}
%\end{figure}
%In Section~\ref{sec:application_examples}, we characterize loss-averse strategies for three important settings: The all-pay auction, facility location on a line segment, and positional scoring rules in voting. Unlike the discrete first-price auction, loss aversion does not induce truthfulness but rather very specific and perhaps surprising behaviors. 
In Section~\ref{sec:vcg_fnp}, we present our main result. We discuss VCG as a combinatorial auction under false-name attacks, when bidders may create shill identities to send bids. It is known that VCG is not dominant strategy truthful in these settings, and previous results were limited in establishing good welfare guarantees for combinatorial auctions generally under false-name attacks. We show that when bidders use DSL strategies, VCG achieves optimal welfare even under the threat of false-name attacks.

\section{DSL: Definition}
\label{sec:motivating_example}

When defining DSL strategies, we take the perspective of a single agent $i$ facing uncertainty. The agent has a utility function $u_i$ that determines her utility given the state of the world, which is comprised of her own action $a_i$, others' actions $a_{-i}$, and agent $i$'s type $\theta_i$. Formally, $u_i(a_i, a_{-i} | \theta_i)$. We denote by $A_i$ the set of all agent $i$'s pure actions, and by $\Delta(A_i)$ the set of all agent $i$'s mixed actions. An action $a_i$ may be from either of these action sets depending on the context. For mixed strategies, $u_i(a_i, a_{-i} | \theta_i) = \Ex[a \sim a_i]{u_i(a, a_{-i} | \theta_i)}$. We denote by $\Theta_i$ the set of all agent $i$'s types. 

\begin{definition} 
\label{def:loss_aversion}

We say that an action $a_i$ of agent $i$ is \textbf{DSL} (given a type $\theta_i$) if for any other action $a'_i$, over the set of outcomes where agent $i$'s utility differs between the actions, the minimal utility attained using $a_i$ is at least as good as that attained by $a'_i$. Formally, let 
$$D_{\theta_i}(a_i, a'_i) = \{a_{-i} \text{\quad s.t. \quad} u_i(a_i, a_{-i} | \theta_i ) \neq u_i(a'_i, a_{-i} | \theta_i)\}.$$ Then, an action $a_i$ is pure/mixed DSL if $\forall a'_i \in A_i$,\footnote{Or, in the mixed case: $\forall a'_i \in \Delta(A_i)$}

$$\min_{\substack{a_{-i} \in D_{\theta_i}(a_i, a'_i)}} u_i(a_i, a_{-i} | \theta_i) \geq \min_{\substack{a_{-i} \in D_{\theta_i}(a_i, a'_i)}} u_i(a'_i, a_{-i} | \theta_i).\footnote{We use the term minimum loosely: When taken over infinite sets that do not have a minimum the definition uses the infimum.}$$

We say that a \textit{strategy} $s_i:\Theta_i \rightarrow A_i$ is pure DSL if it maps any type $\theta_i$ to a corresponding DSL pure action $a_i$. We say that $s_i:\Theta_i \rightarrow \Delta(A_i)$ is mixed DSL if it maps any type $\theta_i$ to a corresponding DSL mixed action $a_i$.

\end{definition}

Notice that in our definition we compare pure strategies only with other pure strategies, i.e., they are DSL with respect to this strategy set. Mixed strategies are DSL w.r.t. all strategies (mixed and pure). We use the term ``nature state'' to mean the actions $a_{-i}$, which may result from either uncertainty over others' types or over their strategic choice: What matters to the agent in the end is what are all of their possible actions. There is seemingly some loss of generality in that we assume that all possible $a_{-i}$ are fixed vectors of actions, and not more generally random variables over actions. But, as we show in the full version of our paper, %Lemma~\ref{prop:nature_settings_deterministic}, 
allowing for the latter loses the usefulness of the DSL notion.

\section{Relations to Prominent Game-theoretic Solution Concepts}
\label{SEC:OTHER_CONCEPTS}

Note: Missing proofs in this section appear in the full version of our paper%Appendix~\ref{app:missing_proofs}
. For completeness, we state the connection of DSL to safety level and what we call Multi-Leximin strategies, although these claims are already established in the literature characterizing the notion of discrimin (see, e.g., \cite{discrimin}). 

\subsection{Dominant Strategy}
\begin{definition}
A weakly dominant action $a_i$ satisfies that for any other action $a'_i$:

(1) For any 
nature state $a_{-i}$, 
$$u_i(a_i, a_{-i} | \theta_i) \geq u_i(a'_i, a_{-i} | \theta_i),$$

and (2) there is such nature state $a_{-i}$ so that the above inequality is strict. 

A weakly dominant strategy is such that maps types to weakly dominant actions. 
\end{definition}

The following result is natural:

\begin{restatable}[]{lemma}{dominantDSL}
\label{lem:dominant_loss_averse}
Every weakly dominant strategy is DSL. 
\end{restatable}

\subsection{Safety Level Strategy and Individual Rationality}

Safety level strategies in non-cooperative games are such strategies that yield a best possible guarantee of utility for a player, without the need to reason about the types or strategies chosen by other players. The example of \cite{aumann1985non} makes a compelling argument for choosing such strategies: There are games where the Nash Equilibrium does not guarantee more than the safety level. In such cases, choosing the equilibrium strategy runs the unnecessary risk of a lower outcome. \cite{tennenholtz2002competitive} extends this insight and shows a class of games where the safety level strategy guarantees a large constant fraction of the Nash equilibrium outcome, without its involved risks. 

Individual rationality is a common requirement in game theory analysis (see, e.g., \cite{roughgarden2010algorithmic}), that requires either that an agent does not participate in a game where it gains negative utility, or that it does not choose a strategy that may yield negative outcomes. We define:
%We formally define the two:
%\moshe{notice that here we ignore the type... perhaps ok...}
\begin{definition}
A \textit{safety level strategy} $s_i$ \cite{aumann1985non} is a strategy (mixed or pure) of player $i$ such that for any type $\theta_i$ it chooses an action $a_i$ so that for any nature state $a_{-i}$ of the other agents, $u_i(a_i, a_{-i} | \theta_i) \geq \max_{a} \min_{a_{-i}} u_i(a, a_{-i} | \theta_i)$. I.e., the strategy guarantees the safety level $L \stackrel{def}{=} \max_{a'} \min_{a'_{-i}} u_i(a', a'_{-i} | \theta_i)$. 

\textit{Individual Rationality} of a strategy $s_i$ of player $i$ satisfies that for any type $\theta_i$ it chooses an action $a_i$ so that for any nature state $a_{-i}$ of the other agents, $u_i(a_i, a_{-i} | \theta_i) \geq 0$. I.e., the strategy guarantees a non-negative utility for the player. 
\end{definition}

The two notions are quite similar, as individual rationality can be seen as a minimal safety level requirement; 
in auctions they are in fact equivalent to a third notion of no over-bidding, under some reasonable conditions (the auction does not charge payments from non-winners, and never charges a winner more than her declared value). We claim:

\begin{restatable}[]{proposition}{DSLsafety}
\label{prop:loss_averse_safety}
A DSL strategy is a safety level strategy, but not necessarily vice-versa.
\end{restatable}

\begin{corollary}
\label{cor:loss_averse_existence}
When there is a finite amount of safety level strategies, and a finite amount of nature states, a DSL strategy is guaranteed to exist.
\end{corollary}

The corollary is a result of Lemma~\ref{lem:leximin_loss_averse} and Lemma~\ref{lem:leximin_existence}. We prove both during our discussion of the lexicographic max-min in the next subsection. 

%\textbf{Open Question:} What results for safety level strategy can be ``upgraded'' to loss aversion? 

\subsection{Lexicographic Max-min}

A very interesting comparison is with another robust solution notion, the lexicographic max-min (also commonly known as leximin). The leximin is especially prevalent in the fair allocation literature, see, e.g., \cite{moulin2004fair}. We consider two possible ways to define it:

\begin{definition}
\textit{Leximin} - Let $U_{a_i}$ be the \textbf{set} of all possible utility outcomes of the action $a_i$ by agent $i$, ordered from small to large, and let $U_{a_i}[j]$ be the $j$ element of $U_{a_i}$ in this ordering. An action $a_i$ lexicographically weakly dominates (LD) another action $a'_i$ if $\min U_{a_i} > \min U_{a'_i}$, or $\min U_{a_i} = \min U_{a'_i}$ and $U_{a_i} \setminus \min U_{a_i}$ LDs $U_{a'_i} \setminus \min U_{a'_i}$ (a recursive definition)%\moshe{s is used? just typo I guess}
. We call an action that LDs all other actions a leximin. A strategy is leximin if it maps all types to leximin actions. 

\textit{Multi-Leximin} - Let $U_{a_i}$ be the \textbf{multiset} of all possible utility outcomes of the action $a_i$ by agent $i$, ordered from small to large. The rest of the definition follows similarly, where importantly in the recursive definition we remove only \textit{one} copy of the minimum element at each step. 

\end{definition}

Note that the (Multi-)leximin notions are only clearly defined when there is a finite amount of nature states $a_{-i}$, otherwise the recursive definition of LD may not terminate. %This is an important advantage of the loss averse strategy definition, as it is not required by it. 
%For simplicity, we assume there is only a finite set of agent $i$ strategies and nature states $a_{-i}$. In examples that do not satisfy these restrictive conditions, it is possible to exploit symmetries and redundancies of the outcomes, and only then apply the leximin notions, but we avoid this as a general discussion. 

We first note that both definitions give stronger notions than safety level strategies. 

\begin{restatable}[]{lemma}{LeximinSafety}
\label{lem:leximin_safety_level}
(Multi-)leximin is a safety level strategy, but not necessarily vice-versa. 
\end{restatable}

%The notion of leximin (rather than multi-leximin) is more common in the literature. 
Despite some similarity in the definition with \text{DSL}, the notion of leximin does not have a special relationship with it: neither implies the other. We demonstrate it using the discrete first-price auction in Example~\ref{ex:leximin_not_loss_aversion} Appendix~\ref{app:application_examples}. 

The notion of multi-leximin is much more closely related to the \text{DSL} notion. In fact, it is a stronger notion:
\begin{restatable}[]{lemma}{leximinDSL}
\label{lem:leximin_loss_averse}
Multi-leximin is a DSL strategy, but not necessarily vice-versa. 
\end{restatable}

\begin{restatable}[]{lemma}{LeximinExistence}
\label{lem:leximin_existence}
When there is a finite amount of safety level strategies, %and a finite amount of nature states, 
multi-leximin is guaranteed to exist.
\end{restatable}

%We recall that we used the above two lemmas to derive Corollary~\ref{cor:loss_averse_existence}. 

%\moshe{did not we say something about leximin can be extended before...} %\yotam{fix with the infinite settings or not} 
An important advantage of the \text{DSL} definition is that it naturally extends to settings with continuous outcomes. It is not clear how to extend the leximin definition to such cases. Thus, one possible way of thinking about the \text{DSL} notion is that it is a somewhat weaker notion of multi-leximin, that can be used in continuous settings, as well as discrete ones.

%It is interesting to compare the outcome of applying min-max regret vs. applying loss aversion to different voting settings. We do that in Subsection~\ref{subsec:voting}. 

%\yotam{Should there also be a subsection about risk-aversion? Not sure what to write there though}

%\subsection{Nash-Equilibrium}
%Note that unlike Nash (that is only guaranteed to exist in the mixed form), it always exists even with only ordinal preferences over the normal form game outcomes. 

%\moshe{we have in figure 1 a proposition about minmax regret I do not see here... need t be careful with these...}
\section{Main Result: Application to VCG under False-name Attacks}
\label{sec:vcg_fnp}

We now move on to present our main result and through it the usefulness of the DSL notion. 
False-name attacks by an agent $i$ in a combinatorial auction are where instead of sending one combinatorial bid, the agent sends multiple combinatorial bids (a vector $\mathbf{b_i}$ rather than a single bid $b_i$). The agent then gets all the items allocated to the ``agents'' (which we call Sybil agents or Sybil bids) $1 \leq j \leq |\mathbf{b_i}|$, and pays the sum of all their payments. Before formally introducing the VCG notations, we note three complexities that are present in our notations: (1) We consider both the notion of \text{DSL} strategies (which has the single agent perspective vs. nature states) and social welfare (which accounts for $n$ different agents). (2) We consider welfare for the real $n$ underlying agents of the auction, but since each may use Sybil identities, the VCG allocations are in terms of the Sybil identities. We allow for both by using sub-indexing. (3) Similar to the case of the first-price auction, discretization of the bid space is essential to the result (a counter-example for continuous VCG appears in the full version of our paper%is in Example~\ref{ex:vcg_underbidding})
. To further simplify the proof, we also assume that the valuation space is discrete, though this assumption can be removed. We allow more granularity to the bid space: valuations are on an $\epsilon$ grid, while bids are on an $\frac{\epsilon}{2|M|!}$ grid. 
%Formally, 

\begin{definition}
$Grid(\epsilon) = \{\epsilon k\}_{k\in \mathcal{N}} = \{ 0, \epsilon, 2\epsilon, \ldots\}$. 
A \textit{combinatorial bid} $b \in B$ over an item set $M$ is a function $b:P(M) \rightarrow Grid(\frac{\epsilon}{2|M|!})$ from the power set of all subsets of $M$ to a non-negative bid value. A combinatorial valuation $v$ is similarly $v:P(M) \rightarrow Grid(\epsilon)$. With the possibility of Sybil attacks, an agent $i$ with valuation (type) $v_i$ sends a vector of bids (action) $\mathbf{b_i} \in B^*$ (i.e., any amount of combinatorial bids), and faces a nature state $\mathbf{b_{-i}} \in B^*$. %The nature state can be written in expansive form as $\mathbf{b}_1, \ldots, \mathbf{b}_{i-1}, \mathbf{b}_{i+1}, \ldots, \mathbf{b}_n$, each in $B^*$. 

Let $\eta_i = |\mathbf{b_i}|, \eta_{-i} = |\mathbf{b_{-i}}|$ be the number of (Sybil) agents in each vector. An allocation $\mathbf{\alpha}^S(\mathbf{b_i}, \mathbf{b_{-i}})$ maps the bid vectors into a partition of $S$ into subsets. We allow indexing $\alpha_1, \ldots, \alpha_n$ to mean the union of items allocated to the Sybil identities of each real agent, as well as sub-indexing $\alpha_{i_1}, \ldots, \alpha_{i_{\eta_i}}$ to mean the items allocated to a specific Sybil identity of agent $i$. 
%\moshe{Overt? a typo? why the first SW is with b??} \yotam{Find a better name (maybe Obs?) and explain why we define and use it}
We denote $SW^{Obs}_{\mathbf{\alpha}} = \sum_{i=1}^n \sum_{j=1}^{\eta_i} b_{i_j}(\alpha_{i_j}), SW^{Real}_{\mathbf{\alpha}} = \sum_{i=1}^n v_i(\alpha_i)$ for the observed social welfare of an allocation as specified in the (possibly Sybil) bids, and the real social welfare of the agents, respectively. We denote $truth_i = v_i$ for the truthful bid.
%$|\eta_i| + |\eta_{-i}|$ subsets, i.e., $a^S_j(\mathbf{b_i}, \mathbf{b_{-i}})$ are the items allocated to agent $j$ (whether $j\in \eta_i$ or $j\in \eta_{-i}$). 

The \textit{VCG} combinatorial auction is the pair of allocation rule $$\mathbf{\alpha}^M(\mathbf{b_i}, \mathbf{b_{-i}}) = \argmax_{\mathbf{\tilde{a}^M(\mathbf{b_i},\mathbf{b_{-i}})}} (SW^{Obs}_{\mathbf{\tilde{a}^M(\mathbf{b_i},\mathbf{b_{-i}})}}),$$ and the payment rule $$p^M_{i_j}(\mathbf{b_i}, \mathbf{b_{-i}}) = SW^{Obs}_{\mathbf{\alpha}^M}(\mathbf{b_i}, \mathbf{b_{-i}}) - SW^{Obs}_{\mathbf{\alpha}^{M \setminus a^M_{i_j}}}(\mathbf{b_i}, \mathbf{b_{-i}}).$$ 
%\moshe{seem missing )...}
Finally, the utility of agent $i$ is $u_i(\mathbf{b_i}, \mathbf{b_{-i}} | v_i) = v_i\left(\bigcup\limits_{1 \leq j \leq \eta_i} a^M_{i_j}(\mathbf{b_i}, \mathbf{b_{-i}})\right) - \sum\limits_{j=1}^{\eta_i} p^M_{i_j}(\mathbf{b_i}, \mathbf{b_{-i}})$. 
\end{definition}

\begin{theorem}
\label{thm:vcg}
When all bidders play DSL strategies, discrete VCG achieves optimal welfare, even under the possibility of false-name attacks and with general valuations. 
\end{theorem}

%\moshe{so can overbid et.? really general? non-monotone?}

\begin{proof}
Our proof follows the following structure: First, we define overbidding Sybil attacks and show that they are not DSL. We then define underbidding attacks and show that they are not DSL. %\moshe{the naive reader might not understand there are remaining attacks... also overbidding and underbidding, do they refer to particular subsets?} 
For any of the remaining attacks, which we call exact-bidding (bidding truthfully is also exact-bidding, but not exclusively so), we show that even though they are not necessarily truthful, they yield maximal welfare. However, this still does not guarantee that one of the remaining strategies is in fact DSL. For this purpose, we show that there exists a DSL strategy: being truthful\footnote{Another, albeit non-constructive method to show there exists a DSL strategy is by showing the finiteness of undominated exact-bidding Sybil attacks, and then use Corollary~\ref{cor:loss_averse_existence}}. 

%(No overbidding)

%Consider some false-name attacker with a combinatorial valuation function $v$. Assume it makes some bids $b_1, \ldots, b_{\eta_i}$%, at least one of which is different than $v$
%. In particular, there must be some set of goods $S$ for which there is a bid $j$ with $b_j(S) \neq v(S)$. 
First, we show that if the Sybil bids are overbidding $v$ (in a sense that will be immediately defined), then, similarly to our proof for the first-price auction (see Appendix~\ref{app:application_examples}), it is not safety level and thus not DSL. This requires slightly more care since the bids are combinatorial and there are several Sybil bids. 
%Denote $a(S) \in Partitions(S)$ for a full allocation of the goods in $S$ to the (Sybil) bidders $b_1, \ldots, b_{\eta_i}$.%, and $\Omega(S)$ the set of all such allocations. 
We say that %\moshe{do we wish to refer to the all possible identities of i here??} 
$\mathbf{b_i} = (b_{i_1}, \ldots,  b_{i_{\eta_i}})$ is overbidding if there is a set $S$ and an allocation $\mathbf{\alpha}^S(\mathbf{b_i})$ so that $\sum_{j=1}^{\eta_i} b_{i_j}(\alpha_{i_j}^S(\mathbf{b_i})) > v_i(S)$. 

\begin{claim}
\label{clm:overbidding}
Overbidding $\implies$ not DSL.
\end{claim}

\begin{proof}
Let us choose a maximizing allocation $\mathbf{\alpha}^S(\mathbf{b_i})$ for $S$. 
%Take a set $S$ (if such a set exists) for which an allocation $a(S)$ exists so that $\sum_{j=1}^{\eta_i} b_j(a_j(S)) > v(S)$, and take $a(S)$ to be the maximizing such allocation for $S$. 
%that has a higher bid for $S$ than $v(S)$. Assume w.l.o.g. $b_1$ is the bid with the highest value for $S$, such that $b_1(S) > v(S)$.

We denote $\bar{b} = \max_{1\leq j\leq \eta_i} \max_{S' \subseteq M} b_{i_j}(S') + v_i(S')$ for a number high enough that if some other agent bids it for any subset of $M$, both the truthful bid $v_i$ or the Sybil attack $\mathbf{b_i}$ will lose that subset. We will use it in our construction of nature states. 
%\yotam{Adding the $v(S')$ to be sure the truthful bidder also doesn't get it, helps the argument in the next paragraph}
By the overbidding condition, we can take the average $\tilde{b} = \frac{v_i(S)}{2} + \frac{1}{2}\sum_{j=1}^{\eta_i} b_{i_j}(a_{i_j}^S)$, so that $\sum_{j=1}^{\eta_i} b_{i_j}(\alpha_{i_j}^S) > \tilde{b} > v_i(S)$. Consider a nature state where the false-name attacker faces exactly one additive bidder $b'$ that has for any good $g \in M \setminus S$, $b'(g) = \bar{b}$, and for any good $g\in S$, $b'(g) = \frac{\tilde{b}}{|S|}$. The optimal observed welfare allocation is to allocate all goods in $M\setminus S$ to $b'$, and allocate the set $S$ as in $\mathbf{\alpha_i}(S)$% \yotam{WHY? WRITE THE SHORT ARGUMENT}
. The payment of bidder $b_i$ must be at least $b'(S) = \tilde{b} > v_i(S)$. Therefore, the attacker has negative utility in this case, while truthfulness is individually rational: i.e., it is not a safety level strategy and so also not DSL.  
\end{proof}

%(No underbidding)
We say that $b_{i_1}, b_{i_{\eta_i}}$ are underbidding if there is a set $S$ so that for any allocation $\mathbf{\alpha}^S \stackrel{def}{=} \mathbf{\alpha}^S(\mathbf{b_i})$ so that $\sum_{j=1}^{\eta_i} b_{i_j}(\alpha_{i_j}^S) < v_i(S)$. %Since there are finitely many such allocations, their maximum is also strictly less than $v_i(S)$. 

\begin{claim}
\label{clm:underbidding}
Underbidding $\implies$ not DSL.
\end{claim}

\begin{proof} Let $\tilde{b} = \frac{1}{2}\sum_{j=1}^{\eta_i} b_{i_j}(\alpha_{i_j}^S) + \frac{v_i(S)}{2}$, then
$$\sum_{j=1}^{\eta_i} b_{i_j}(\alpha_{i_j}^S(\mathbf{b_i})) < \tilde{b} < v_i(S).$$

Let $b'$ be constructed as in the overbidding case. The allocation $\mathbf{\alpha}^M(\mathbf{b_i}, b')$ allocates no items to the Sybil bidders of agent $i$. However, the allocation given agent $i$ bids truthfully $\mathbf{\alpha}^M(truth_i, b')$, allocates the set $S$ to her with payment $\tilde{b}$, which yields agent $i$ a positive utility $v_i(S) - \tilde{b}$. This yields 
$$\min_{\mathbf{b_{-i}} \in D_{v_i}(\mathbf{b_i}, truth_i)} u_i(\mathbf{b_i}, \mathbf{b_{-i}} | v_i) = 0.$$
On the other hand, we claim that since we know DSL strategies are not overbidding, there are no nature states for which an underbidding Sybil attack gets positive utility while bidding truthfully gets $0$ utility. Assume towards contradiction $truth_i$ gets $0$ utility. It then either does not win any item, or wins some set $S$ and pays $v_i(S)$ for it. Let $S$ be the set that the Sybil bidders win to gain positive utility. As there is no overbidding, this set can be won by $truth_i$ as well (in the respective maximizing allocation)\footnote{In full generality, $truth_i$ may win a set $s$ that has partial intersection with $S$. The analysis of this case is essentially the same, and stems from the fact that the alternative value for the items forces zero utility on the truthful agent.}. Then, %\yotam{The point here is to establish the next equation: It's established no matter which set $truth_i$ wins... think how to explain it.}
%\begin{equation}
\begin{align*}
    %& SW_{\alpha^{M \setminus S}(\mathbf{b_i},\mathbf{b_{-i}})} = \qquad \text{(Additional bids)}\\
    & SW^{Obs}_{\alpha^{M \setminus S}}(\mathbf{b_i}, \mathbf{b_{-i}})  && \text{(No overbidding)}\\ 
    & = SW^{Obs}_{\alpha^{M \setminus S}}(truth_i, \mathbf{b_{-i}})  && \text{($i$ wins only $S$)}\\ 
    & = SW^{Obs}_{\alpha^{M \setminus S}}(\mathbf{b_{-i}}) && \text{($i$'s truthful payment)}\\
    & = SW^{Obs}_{\alpha^{M}}(\mathbf{b_{-i}}) - v_i(S) && \text{(More bids)}\\
    & \leq SW^{Obs}_{\alpha^{M}}(\mathbf{b_i},\mathbf{b_{-i}}) - v_i(S)%\geq \qquad \text{(No overbidding)} \\
    %& SW^{Obs}_{\alpha^{M}(\mathbf{b_i}, \mathbf{b_{-i}})} - v_i(S) \\ 
\end{align*}
%\end{equation}

So 
\begin{equation}
\label{eq:value_of_S}
v_i(S) \leq SW^{Obs}_{\alpha^{M}(\mathbf{b_i},\mathbf{b_{-i}})} - SW^{Obs}_{\alpha^{M \setminus S}(\mathbf{b_i}, \mathbf{b_{-i}})} \end{equation}

Since our choice of $S$ assumes the Sybil bids win exactly it, we have
$$SW^{Obs}_{\alpha^{M \setminus S}(\mathbf{b_i}, \mathbf{b_{-i}})} + SW^{Obs}_{\alpha^{S}(\mathbf{b_i})} = SW^{Obs}_{\alpha^{M}(\mathbf{b_i},\mathbf{b_{-i}})},$$
and so, together with Eq.~\ref{eq:value_of_S},

$$v_i(S) \leq SW^{Obs}_{\alpha^{M}(\mathbf{b_i},\mathbf{b_{-i}})} - SW^{Obs}_{\alpha^{M \setminus S}(\mathbf{b_i}, \mathbf{b_{-i}})} = SW^{Obs}_{\alpha^{S}(\mathbf{b_i})}.$$

Since there is no overbidding, $SW^{Obs}_{\alpha^{S}(\mathbf{b_i})} = v_i(S)$. 
%i.e., together with no overbidding, for the Sybil bids to win $S$ it holds that $\sum_{j=1}^{\eta_i} b_{i_j}(\alpha_{i_j}) = v_i(S)$. 

We now show that any Sybil bidder $j$ pays $v_i(S) - \sum_{1 \leq t \neq j \leq \eta_i} b_{i_j}(\alpha_{i_j})$. 
%\moshe{missing equation...}
Since $S$ is allocated to the Sybil bidders and $M\setminus S$ to others, 
\begin{equation}
\label{eq:payment_underbidding}
\begin{split}
    & SW^{Obs}_{\mathbf{\alpha}^{M \setminus a^M_{i_j}}}(\mathbf{b_i}, \mathbf{b_{-i}}) =  \\
    & SW^{Obs}_{\mathbf{\alpha}^{M \setminus S}}(\mathbf{b_i}, \mathbf{b_{-i}}) + SW^{Obs}_{\mathbf{\alpha}^{S \setminus a^M_{i_j}}}(\mathbf{b_i}, \mathbf{b_{-i}})
\end{split}
\end{equation}

Then,
\begin{align*}
    & p^M_{i_j} = SW^{Obs}_{\mathbf{\alpha}^M}(\mathbf{b_i}, \mathbf{b_{-i}}) - SW^{Obs}_{\mathbf{\alpha}^{M \setminus a^M_{i_j}}}(\mathbf{b_i}, \mathbf{b_{-i}}) && %\text{(Eq.~\ref{eq:payment_underbidding})}
    \\ 
    & = SW^{Obs}_{\mathbf{\alpha}^M}(\mathbf{b_i}, \mathbf{b_{-i}}) - SW^{Obs}_{\mathbf{\alpha}^{M \setminus S}}(\mathbf{b_i}, \mathbf{b_{-i}})\\
    & \qquad \qquad \qquad \qquad - SW^{Obs}_{\mathbf{\alpha}^{S \setminus a^M_{i_j}}}(\mathbf{b_i}, \mathbf{b_{-i}}) &&  \\
    & = v_i(S) - SW^{Obs}_{\mathbf{\alpha}^{S \setminus a^M_{i_j}}}(\mathbf{b_i}, \mathbf{b_{-i}}) \\
    & = v_i(S) - \sum_{j=1}^{\eta_i} b_{i_j}(\alpha_{i_j})
    \end{align*}

%\yotam{explain why it is determined only by $\mathbf{b_i}$}
The total payment of agent $i$ is then \[
\begin{split}
& \sum_{j=1}^{\eta_i} p^M_{i_j} = \sum_{j=1}^{\eta_i} v_i(S) - \sum_{1 \leq t \neq j \leq \eta_i} b_{i_j}(\alpha_{i_j}) = \\
& \eta_i \cdot v_i(S) - (\eta_i - 1) \sum_{j=1}^{\eta_i} b_{i_j}(\alpha_{i_j}) = v_i(S).
\end{split}
\]
This concludes that whenever the utility of $truth_i$ is $0$, then the utility for the Sybil attack is $0$ as well. 

%\yotam{THIS IS THE PART THAT REALLY NEEDS SOME WORK}
%So, we conclude that for $\mathbf{b_{-i}}$, the total value of $M \setminus S$ is the value of $M$ minus $v_i(S)$, i.e., together with no overbidding, for the Sybil bids to win $S$ it holds that $\sum_{j=1}^{\eta_i} b_{i_j}(\alpha_{i_j}) = v_i(S)$. For any Sybil bidder $j$ the price they pay is then $v_i(S) - \sum_{1 \leq t \neq j \leq \eta_i} b_{i_j}(\alpha_{i_j})$, and so the total payment is $\sum_{j=1}^{\eta_i} v_i(S) - \sum_{1 \leq t \neq j \leq \eta_i} b_{i_j}(\alpha_{i_j}) = \eta_i \cdot v_i(S) - (\eta_i - 1) \sum_{j=1}^{\eta_i} b_{i_j}(\alpha_{i_j}) = v_i(S).$ 
%This concludes that the utility for the Sybil attack is $0$ as well. 

In any other case, the utility of $truth_i$ must be strictly positive, and since the bids are discrete the minimum over all these cases satisfies 
$$\min_{\mathbf{b_{-i}} \in D_{v_i}(\mathbf{b_i}, truth_i)} u_i(truth_i, \mathbf{b_{-i}} | v_i) \geq \frac{1}{2|M|!}.$$
Therefore, underbidding is not DSL. 
\end{proof}

We consider \textit{exact-bidding} such Sybil bids that have for any set of items $S$, $\max_{\mathbf{\alpha}^S(\mathbf{b_i})} \sum_{j=1}^{\eta_i} b_{i_j}(\alpha_{i_j}(S)) = v_i(S)$. These are exactly all the Sybil attacks that are neither overbidding nor underbidding. $truth_i$ is also exact-bidding.
\begin{claim}
\label{clm:exact-bidding}
Exact-bidding $\implies$ optimal welfare.
\end{claim}

\begin{proof}Consider an allocation $\mathbf{\alpha}_F \stackrel{def}{=} \mathbf{\alpha}^M(\mathbf{b_i},\mathbf{b_{-i}})$ attained when all players choose an exact-bidding attack, vs $\mathbf{\alpha}_T \stackrel{def}{=} \mathbf{\alpha}^M(truth_i,\mathbf{truth_{-i}})$. We have
\[
\begin{split}
    & SW_{\mathbf{\alpha}_T}^{Real} \leq \qquad \text{(Truthful)} \\
    & SW_{\mathbf{\alpha}_T}^{Obs} \leq \qquad \text{(No underbidding)} \\
    & SW_{\mathbf{\alpha}_F}^{Obs} \leq \qquad \text{(No overbidding)} \\
    & SW_{\mathbf{\alpha}_F}^{Real}
\end{split}
\]

In words, since there is no underbidding in the Sybil attack, if we take the set allocated to each agent $i$ under the allocation that maximizes welfare under truthfulness, there are Sybil bidders $i_{j_1}, \ldots, i_{j_k}$ with the same aggregate valuation for it. So, $SW_{\mathbf{\alpha}_F}^{Obs}$ is lower bounded by the optimal truthful welfare. Since there is also no overbidding, whatever allocation is chosen as $\mathbf{\alpha}_F^{Obs}$ is at least as good to each agent $i$ as is declared. 
\end{proof}

\begin{claim}
\label{clm:truth_loss_averse_sybil}
$truth_i$ is DSL.
\end{claim}

\begin{proof}Consider some exact-bidding Sybil attack $\mathbf{b_i}$.

Case 1: There is a set $S$ so that $\forall 1\leq j \leq \eta_i, b_{i_j}(S) < v_i(S)$. Then, by the exact-bidding condition there must be some allocation $\mathbf{\alpha}^S(\mathbf{b_i})$ (with at least two non-empty allocations $\alpha^S_{i_j}$) so that $\max_{j=1}^{\eta_i} b_{i_j}(\alpha^S_{i_j}) < \sum_{j=1}^{\eta_i} b_{i_j}(\alpha^S_{i_j}) = v_i(S)$. Consider the nature state where there is one bid $b'$ so that $b'(\alpha^S_{i_j}) = v_i(\alpha^S_{i_j})$ for any $1\leq j \leq \eta_i$, and the rest of the sets are defined upward-monotonely: They inherit the largest value of a subset. With this nature state, the Sybil attack has utility $0$. On the other hand, $truth_i$ has positive utility of $v_i(S) - \max_{j=1}^{\eta_i} v_i(\alpha_{i_j}) > 0$. Since $truth_i$ is individually rational, it is thus DSL w.r.t. such Sybil attacks. 

Case 2: For every set $S$, there is such $j'$ with $b_{i_{j'}}(S) = v_i(S)$. It must hold by the exact-bidding condition that for any allocation $\mathbf{\alpha}^S(\mathbf{b_i})$, $\sum_{j=1}^{\eta_i} b_{i_j}(\alpha^S_{i_j}) \leq v_i(S) = b_{i_{j'}}(S)$. We may assume that VCG prefers to assign larger bundles when tie-breaking between possible assignments. Then, it must be that any allocation to the Sybil bidders is given to one Sybil bidder as a whole bundle. It is then weakly better to send only $b_{i_{j'}}$ as a single bid instead of $\mathbf{b_i}$. Furthermore, it is then weakly better to send $truth_i$, since truthfulness is dominant for single bid VCG. Since this is true given any nature state, the Sybil attack is weakly dominated by $truth_i$, which implies $truth_i$ is DSL with respect to it. 

This covers all the exact-bidding Sybil attacks. \text{DSL} strategies with respect to overbidding and underbidding attacks are implied by the relevant discussion. Overall this covers all Sybil attacks. 
\end{proof}

\end{proof}

\vspace{-1cm}
\section{Discussion and Future Directions}

In the example of the discrete first-price auction in Section~\ref{sec:motivating_example}, as well as in our main result in Section~\ref{sec:vcg_fnp}, the DSL solution concept leads to optimal results: truthfulness (or near truthfulness), and optimal revenue or welfare. In Appendix~\ref{app:application_examples} we study the first-price auction, and show similar results for its discrete variant. However, for the classic setting of voting, we show in the full version of this paper that this is not the case, and that solutions may have various surprising forms. 

A robust notion missing from our discussion in Section~\ref{SEC:OTHER_CONCEPTS} is min-max regret. We show in Appendix~\ref{app:minmax_regret} it does not imply or is implied by our notion of \text{DSL}, and give further characteristics of it. It is also compared with our notion as part of our discussion of the discrete first-price auction in Appendix~\ref{app:application_examples}. %\moshe{I miss it there?}

In our definition of \text{DSL}, we consider only pure nature states. We justify this choice in Appendix~%\ref{app:randomized_env} 
C of the full version, by showing that if we consider mixed nature states as well, then the \text{DSL} and safety level notions become one. In Appendix~%\ref{app:loss_averse*}
D of the full version, we show a possible refinement of our notion of \text{DSL}, and demonstrate why it may be useful. 

A few immediate open questions follow our work:
\begin{itemize}
    \item We find that \text{DSL} is a stronger notion than safety level. In settings previously studied that proved performance guarantees for safety level strategies, do DSL strategies exist? Can they yield better performance guarantees?
    
    \item In the case of single-item auctions, our analysis of the discrete first price auction implies that with DSL bidders, it is possible to achieve optimal welfare and revenue. Does this extend to combinatorial auctions? If so, does it hold even when the discretization must be polynomially bound? 

    %\item For a designer or mediator that wishes to influence the outcome of a normal-form game where agents play DSL strategies (similar to the approaches presented in \cite{kimplementation, incentiveEngineeringBoolean}
    
    \item In the presence of partial knowledge or the option to elicitate it (similar to the ideas in \cite{lu2011robust}), what would the DSL action be? This is relevant, for example, when agents arrive sequentially, and so the set of feasible nature states diminishes for later agents.  
\end{itemize}

\section*{Acknowledgements}

Yotam Gafni and Moshe Tennenholtz were supported by the European Research Council (ERC) under the European Union’s Horizon 2020 research and innovation programme (Grant No. 740435). 

%Another interesting direction is ... \yotam{attempting to resolve impossiblity results for blockchain mechanism design and credible auctions}.
%We now show that this behavior of additive bidders comes with a welfare cost (i.e., one could hope that additive bidders' deviation from truthfulness does not result in welfare loss: However this is not the case). 

%\begin{example}
%Let $v(a) = v(b) = 100, v({a,b}) = 200$, then bidding truthfully is dominated by bidding the two split Sybil bids. 
%\end{example}
%\yotam{Write details}

%Gibbard-Saterthwaite and other voting / social choice settings

%\textbf{Watch foundations of non-truthful mechanisms of Jason Hartline}

%\textbf{Are there Bayesian Incentive Compatibility cases that can fit this new framework?}

%\textbf{Compare structure of paper with obviously SP and Ehud Kalai talk in NoamFest}

\bibliographystyle{eptcs}
\bibliography{ref.bib}

%\clearpage

\appendix

\section{Min-max Regret}
\label{app:minmax_regret}
Another robust solution notion is the min-max regret \cite{savage1951theory}. The notion has many uses in voting: \cite{ferejohn1974paradox} showed it can be used to explain why voters choose to participate in elections, and \cite{merrill1981strategic} used it to ``resolve'' the Gibbard-Satterthwaite impossibility theorem (see, e.g., \cite{brandt2016handbook}), by showing that plurality voting (for example) is truthful under this notion. \cite{lu2011robust} showed how when only partial preferences are known, voting rules can use this notion to decide a winner, and design good elicitation schemes. 

\begin{definition}
\label{def:minmax_regret}
\textit{Regret} for an action $a_i$ and nature state $a_{-i}$ given a type $\theta_i$ is $$Reg(a_i, a_{-i} | \theta_i) = \max_{a'_i} u(a'_i, a_{-i} | \theta_i) - u(a_i, a_{-i} | \theta_i).$$

\textit{Max regret} for an action $a_i$ given a type $\theta_i$ is
$$Reg(a_i, a_{-i} | \theta_i).$$ 
A min-max regret action belongs to %\moshe{notice it is not uniquely defined. i.e. a set...}
$$\argmin_{a_i} \max_{a_{-i}, a'_i} u(a'_i, a_{-i} | \theta_i) - u(a_i, a_{-i} | \theta_i).$$ 
\end{definition}

In words, the regret of an action $a_i$ under nature state $a_{-i}$ is the maximal lost utility $u(a'_i, a_{-i} | \theta_i) - u(a_i, a_{-i} | \theta_i)$ of choosing $a_i$ instead of $a'_i$, over all possible actions $a'_i$ (this regret is non-negative, as there is always the option of choosing $a_i$ itself). Max regret is the maximal such regret over all nature states, and the min-max regret action is the action $a_i$ that has minimal max regret.

\begin{proposition}
\label{prop:dominant_maxminreg}
Dominant strategy $\implies$ min-max regret
\end{proposition}
\begin{proof}
Consider a dominant strategy $s$, fix a type $\theta_i$, and let $a = s(\theta_i)$. For any $a', a_{-i}$, we have that $u(a', a_{-i}) - u(a,a_{-i}) \leq 0$, i.e., the max regret for $a$ is $0$, the minimum possible, and so $a$ is a min-max regret action. Since this holds for all types, $s$ is a min-max regret strategy. 
\end{proof}

\begin{example}
\label{ex:minmaxreg_safety}
\textit{Min-max regret $\centernot \implies$ safety level}

Consider two players, with actions $a,b$, and $A,B$ respectively. Consider $u_1(a,A) = u_1(a,B) = 0, u_1(b,A) = -1, u_1(b,B) = 100$. The max regret of $a$ for player $1$ is 100, and the max regret of $b$ is $1$, and so $b$ is the min-max regret strategy, while $a$ is the unique safety level strategy. 
\end{example}

\section{DSL Strategies: Application to the First-price Auction}
\label{app:application_examples}

\subsection{The First-price Auction}
As an illustrative example, 
%starting in Section~\ref{sec:motivating_example}, 
we demonstrate the usage of our solution concept using the first price and discrete first price single item auctions. %\moshe{unclear where this comes from?} 
Interestingly, we show that in the first-price auction, there are no DSL strategies. However, moving to a discrete setting, we show that in the discrete first-price auction \cite{chwe1989discrete}, there is a unique DSL strategy, which achieves maximal welfare and near maximal revenue. 
%We first formulate a discretization of the first-price auction bid space, in the spirit of \cite{chwe1989discrete}, and name it the $\epsilon$-first-price auction (where $\epsilon$ describes the granularization of the bid space):

\begin{definition}
An agent $i$ has a value (type) $v_i$ for an item. 
The agent's bid $b_i$ (action) and nature states $b_{-i}$ are from the same bid space. 
The auctioneer allocates the item to the highest bidder (either the agent or nature, tie-breaking towards nature) and if the agent wins it receives $v_i - b_i$, and otherwise $0$.

\textit{First-price auction (FPA)}: bid space is $0 \leq b_i \leq v_i$. 

\textit{Discrete first-price auction (DFPA)}: bid space is $b_i \in \{\epsilon \cdot k | \epsilon \cdot k \leq v_i \}_{k\in \mathcal{N}}$. 
\end{definition}

Note that we reformulate the auctions to suits our agent perspective formulation. Moreover, we omit strategies that are not individually rational (in the (discrete) first-price auction, %\moshe{may be zero too?} 
overbidding has negative utility in some nature states), which is justified by our later discussion in Proposition~\ref{prop:loss_averse_safety}. 
We also ignore multitude in nature states that does not change the auction outcome. I.e., we only consider the highest bids by others as the nature state, and not the entire bid vector. For the DFPA, we denote $\epsilon_{net}(v_i) = \epsilon \cdot \max_{\substack{\epsilon n \leq v_i}} n$, i.e., the closest possible bid below the agent's value of the item. 

In the first-price auction, the notion of DSL strategies is not of much help:

\begin{lemma}
\label{lem:first_price}
In the first-price auction, there are no DSL bid strategies.
\end{lemma}

\begin{proof}

First, consider some bid $0 \leq b_i < v_i$. Compare it with another bid $b'_i$ that satisfies $b_i < b'_i < v_i$.\footnote{Note that the proof is written for the pure DSL case. However, it immediately generalizes to the mixed case, by adapting ``$b_i < v_i$'' to ``has a positive probability to satisfy $b_i < v_i$'', etc.}  Consider a nature state $b_{-i}$ so that $b_i < b_{-i} < b'_i$. Then, $0 = u_i(b_i, b_{-i} | v_i) \neq u_i(b'_i, b_{-i} | v_i) = v_i - b'_i$. Thus, $$\min\limits_{b_{-i} \in D_{v_i}(b_i, b'_i)} u_i(b_i, b_{-i} | v_i) = 0.$$ 

On the other hand, for the bid $b'_i$ and for some nature state $b_{-i}$, $u_i(b'_i, b_{-i}) = 0$ if and only if $b_{-i} \geq b'_i$. In all such cases, it also holds that $u_i(b_i, b_{-i} | v_i) = 0$. In all other cases, i.e., when $b_{-i} < b'_i$, the utility of the bidder satisfies $u_i(b'_i, b_{-i} | v_i) = v_i - b'_i$. We conclude that \[
\begin{split}
    & \min\limits_{b_{-i} \in D_{v_i}(b_i, b'_i)} u_i(b'_i, b_{-i} | v_i) = v_i - b'_i \\
     > & \min\limits_{b_{-i} \in D_{v_i}(b_i, b'_i)} u_i(b_i, b_{-i} | v_i) = 0,
    \end{split}
    \]
    and the bid strategy $b_i$ is not DSL. 

If $b_i = v_i$, then for any nature state $b_{-i}$, $u_i(b_i, b_{-i} | v_i) = 0$. For some $0 \leq b'_i < b_i$, for any nature state $b_{-i}$ where its utility is non-zero, we have $u_i(b'_i, b_{-i} | v_i) = v_i - b'_i > 0$, and so similarly to before $b_i = v_i$ is not DSL. 
\end{proof}

However, things get more interesting with the DFPA:

\begin{lemma}
\label{lem:discrete_first_price}
In the discrete first-price auction: 

For types that have $\epsilon_{net}(v_i) \neq v_i$, and types with $\epsilon_{net}(v_i) = v_i = 0$, bidding $\epsilon_{net}(v_i)$ is the unique DSL bid.

For types with $\epsilon_{net}(v_i) = v_i \neq 0$, the unique DSL bid is $\epsilon_{net}(v_i) - \epsilon$. 
\end{lemma}

We first give a proof for the pure DSL case.

\begin{proof}
The argument why any other bid strategy is not DSL follows a discretized version of the proof for Lemma~\ref{lem:first_price}. 

Case 1: $\epsilon_{net}(v_i) \neq v_i$

%Since in this case $\epsilon_{net}(v_i) < v_i$, it is individually rational, and so since any bid with $b'_i > v_i$ can have a negative utility outcome, the loss aversion condition is satisfied with respect to these bids. 
Consider some bid with $0 \leq b'_i < \epsilon_{net}(v_i)$. By the same argument as in the first part of the proof of Lemma~\ref{lem:first_price}, bidding $\epsilon_{net}(v_i)$ is \text{DSL} w.r.t. $b'_i$. Since there are no bids with $\epsilon_{net}(v_i) < b'_i \leq v_i$ by the definition of $\epsilon_{net}$, we conclude that $\epsilon_{net}(v_i)$ is DSL w.r.t. all other bids, i.e., DSL. 

Case 2: $\epsilon_{net}(v_i) = v_i = 0$

The unique safety level bid is to bid $0$, and so by Proposition~\ref{prop:loss_averse_safety} it is also the unique DSL strategy.

Case 3: $\epsilon_{net}(v_i) = v_i \neq 0$ 

Similar to the first case, with the difference that bidding $v_i$ always leads to utility $0$, and so the DSL bid bracket is $v_i - \epsilon$. 
\end{proof}

The following lemma completes the mixed DSL case:
\begin{lemma}
\label{lem:mixed_loss_averse_discrete}
In the discrete first-price auction with mixed strategies, following the unique DSL pure strategy is the unique DSL strategy. 
\end{lemma}
\begin{proof}
We show the proof for case 1 where $\epsilon_{net}(v_i) \neq v_i$. The other cases are done similarly. 

Let $s_i$ be the stated strategy, $v_i$ the valuation (type) and the bid $b = s_i(\theta_i)$. Let $b'$ be some other bid: since $b \neq b'$, the bracket $\epsilon_{net}(v_i)$ has probability $p < 1$ of being the actualized bid. Consider the case $b_{-i}$ where another bidder bids $\epsilon_{net}(v) - \epsilon$, and ties are broken in favor of the other bidder. Then, $u_i(b', b_{-i} | v_i) = \Ex[\tilde{b}' \sim b']{u_i(\tilde{b}', b_{-i} | v_i)} = p \cdot (v_i - \epsilon_{net}(v_i)) + (1-p) \mathbbm{1}[\tilde{b}' > \epsilon_{net}(v_i)] \cdot (v_i - \tilde{b}') = p \cdot (v_i - \epsilon_{net}(v_i)) + (1-p) \mathbbm{1}[\tilde{b}' > v_i] (v_i - \tilde{b}') < p \cdot (v_i - \epsilon_{net}(v_i)) < v_i - \epsilon_{net}(v_i)$. In any nature state and actualized outcome over the mixed bid $b'$, if $b$ does not win the item, then $b'$ does not win the item, or, alternatively, it wins it and receives negative utility. So, $\min_{b_{-i} \in D_{v_i}(b, b')} u_i(b, b_{-i} | v_i) \geq v_i - \epsilon_{net}(v_i) > \min_{b_{-i} \in D_{v_i}(b, b')} u(b', b_{-i} | v_i)$, and so by the DSL condition $b'$ is not DSL (and $b$ is DSL w.r.t. $b'$). 

\end{proof}

The simple intuition as to why the discrete first-price auction ``works'' (to guarantee a DSL strategy) and the first-price auction does not, is that in the first-price auction there is always a ``safer'' bid that would guarantee winning the item in more nature states. In the discrete first-price auction, due to bracketing, the highest bracket that can have positive utility is that DSL bid. Note that this is ``almost'' truthful: When $\epsilon_{net}(v_i) \neq v_i$, it is the closest bracket to $v_i$, and it is less than $\epsilon$ away from it. When $\epsilon_{net}(v_i) = v_i$ (which should be seen as a rare case, where the value precisely matches the epsilon net), it is not the truthful bracket, but it is $\epsilon$ close to it. It is also very close to optimal revenue for the auctioneer: If $n$ individually rational agents participate, the most the auctioneer can get is $\max_{1\leq i\leq n} v_i$. If they play DSL strategies, she will get at least $\max_{1\leq i\leq n} v_i - \epsilon$.

We note that for the discrete first-price auction, DSL identifies with multi-leximin. 
\begin{corollary}
The unique DSL strategy of the discrete first-price auction is also the unique multi-leximin strategy.
\end{corollary}
\begin{proof}
For an agent $i$ with value $v_i$ there is a finite amount of safety level strategies, namely all the strategies with $b_i \leq v_i$, the amount of which is at most $\lceil \frac{v_i}{\epsilon} \rceil + 1$. %We previously showed how it is possible to consider only a finite amount of nature states in such case (see the discussion in Example~\ref{ex:leximin_not_loss_aversion}).  
By Lemma~\ref{lem:leximin_existence}, there must exist a multi-leximin strategy. By Lemma~\ref{lem:leximin_loss_averse} it is also DSL. Since there is a unique DSL strategy by Lemma~\ref{lem:discrete_first_price}, it must also be the unique multi-leximin.
\end{proof}

On the other hand, we now see that min-max regret yields a different solution to the discrete first-price auction than \text{DSL}, i.e., the two notions do not imply each other. 
\cite{tennenholtz2001rational} previously applied min-max regret in auction settings, and in particular discussed the DFPA in their Claim 3.1, which we restate adapted to our notations: %\moshe{claim where? so we do not have our own results for minmax regret?}
\begin{claim}
In the discrete first-price auction, the min-max regret strategy is to bid $\epsilon_{net}(\frac{v_i}{2})$. 
\end{claim}
\begin{proof}
For any bid $b_i$, the maximum regret is either $b_i$ itself (in the case when no other bidders show up and it was possible to bid and pay $0$), or $v_i - (b_i + \epsilon)$ (in the case when another bidder bids $b_i$ and the item goes to her.\footnote{This is true under worst-case arbitrary tie-breaking. If tie-breaking is uniformly random between bidders of the same bracket, this is still true as the limiting regret when there are $n\rightarrow \infty$ bidders in the same bracket} We are thus looking for $\argmin_{b_i} \max \{b_i, v_i - b_i - \epsilon \}$, among the $\epsilon_{net}$ feasible bids. 

For $b'_i < \epsilon_{net}(\frac{v_i}{2})$, the regret is thus at least 
\[
\begin{split}
& Reg(b'_i) \geq v_i - b_i - \epsilon \\
& \geq v_i - (\epsilon_{net}(\frac{v_i}{2}) - \epsilon) - \epsilon = v_i - \epsilon_{net}(\frac{v_i}{2}) \\
& \geq \max \{v_i - \epsilon_{net}(\frac{v_i}{2}) - \epsilon, \frac{v_i}{2}\} \\
& \geq \max \{v_i - \epsilon_{net}(\frac{v_i}{2}) - \epsilon, \epsilon_{net}(\frac{v_i}{2})\} = Reg(b_i).
\end{split}
\]

For $b'_i > \epsilon_{net}(\frac{v_i}{2})$, the regret is at least 
\[
\begin{split}
    & Reg(b'_i) \geq b'_i \geq \epsilon_{net}(\frac{v_i}{2}) + \epsilon \\
    & \geq \max \{ \epsilon_{net}(\frac{v_i}{2}), \frac{v_i}{2}\} \\
    & \geq \max \{ \epsilon_{net}(\frac{v_i}{2}), v_i - \epsilon_{net}(\frac{v_i}{2}) - \epsilon\} = Reg(b_i).
\end{split}
\]

We conclude that $\epsilon_{net}(\frac{v_i}{2})$ is the min-max regret bid strategy. 
\end{proof}

Finally, we use the discrete first-price auction to demonstrate the difference between leximin and DSL strategies. 
\begin{example}
\label{ex:leximin_not_loss_aversion}
We demonstrate that leximin is different from \text{DSL} using the discrete first-price auction. %To have a finite set of outcomes, we fix a number of $n$ bidders, and for a bidder $i$ with value $v_i$ we consider only her individually rational strategies (as both the leximin and loss-averse strategies would be a subset of that space). Then, we can consider all bids by the other agents that are higher than $v_i$ as the same action, as the implication for agent $i$'s outcome is the same. This does not change the result for the loss-averse strategy: It is still as in Lemma~\ref{lem:discrete_first_price}. However, 
Bidding $0$ is the leximin action, as its set of outcomes is simply the set of two items $U_0 = \{0, v_i\}$: This is the leximin since any other bid $b_i > 0$ has $U_{b_i} = \{0, v_i - b_i\}$. %(if the bidder does not win the item), and also has some other outcome lower than $v_i$ (namely, $v_i - b_i$). 
\end{example}

\end{document}